\documentclass[aps,prl,twocolumn,showpacs,superscriptaddress]{revtex4}
\usepackage{amsmath}
\usepackage{amsfonts}
\usepackage{amssymb}
\usepackage{amsthm}
\usepackage{graphicx}
\usepackage{dcolumn}
\usepackage{bm}
\usepackage{color}

\newcommand{\ket}[1]{\ensuremath{\left|{#1}\right\rangle}}

\newtheorem{theorem}{Theorem}
\newtheorem{lemma}{Lemma}

\newtheorem{fact}{Fact}
\theoremstyle{definition}

\newtheorem{example}{Example}
\newtheorem{problem}{Problem}

\theoremstyle{remark}

\begin{document}

\title{Existence of Universal Entangler}

\author{Jianxin Chen}
\email{chenkenshin@gmail.com}
\affiliation{ State Key Laboratory of Intelligent Technology and
Systems, Department of Computer Science and Technology, Tsinghua
University, Beijing 100084, China }
\author{Runyao Duan}
\email{dry@tsinghua.edu.cn}
\affiliation{ State Key Laboratory of Intelligent Technology and
Systems, Department of Computer Science and Technology, Tsinghua
University, Beijing 100084, China }
\author{Zhengfeng Ji}
\email{jizhengfeng98@mails.thu.edu.cn}
\affiliation{ State Key Laboratory of Intelligent Technology and
Systems, Department of Computer Science and Technology, Tsinghua
University, Beijing 100084, China }
\author{Mingsheng Ying}
\email{yingmsh@tsinghua.edu.cn}
\affiliation{ State Key Laboratory of Intelligent Technology and
Systems, Department of Computer Science and Technology, Tsinghua
University, Beijing 100084, China }
\author{Jun Yu}
\email{majunyu@ust.hk} \affiliation{ Department of Mathematics, Hong
Kong University of Science and Technology, Clear Water Bay, Kowloon,
Hong Kong }

\date{\today}

\begin{abstract}
A gate is called an entangler if it transforms some (pure) product
states to entangled states. A universal entangler is a gate which
transforms all product states to entangled states. In practice, a
universal entangler is a very powerful device for generating
entanglements, and thus provides important physical resources for
accomplishing many tasks in quantum computing and quantum
information. This Letter demonstrates that a universal entangler always
exists except for a degenerated case. Nevertheless, the problem how
to find a universal entangler remains open.
\end{abstract}


\maketitle

\paragraph{Introduction.}

It is a common sense in the quantum computation and quantum
information community that entanglement is an extremely important
kind of physical resources. How to generate this kind of resources therefore becomes an important problem. One possibility one may naturally conceive is to generate entanglement from
product states which can be prepared spatially separately. If a gate
can transform some product states to entangled states, then we call
it an entangler. For some given product states, it is
not difficult to find an entangler that maps them to entangled states. Here we consider a more challenging problem: Does there exist some entangler that can transform \textit{all} product states to entangled states. We call such an
entangler a \textit{universal} entangler~\cite{Buzek00,Buzek04}. One may even wonder at the existence of
universal entanglers.

Formally, the problem of existence of universal entanglers in
bipartite systems can be stated as follows:

\begin{problem}\label{or-prob}
Suppose Alice and Bob have two quantum systems with state spaces $\mathcal{H^A}$ and $\mathcal{H^B}$ respectively. Whether there exists a unitary operator $U$ acting on $\mathcal{H^A\otimes H^B}$ such that $U(\ket{\phi}\otimes \ket{\psi})$ is always entangled for any $\ket{\phi}\in \mathcal{H^A}$ and $\ket{\psi}\in \mathcal{H^B}$?
\end{problem}

The purpose of this letter is to present a complete solution to the above problem. We have our main theorem as follows:

\begin{theorem}\label{main}
Given a bipartite quantum system $\mathcal{H_m\otimes H_n}$, where $\mathcal{H_m}$ and $\mathcal{H_n}$ are Hilbert space with dimension m and n respectively, then there exists a unitary operator U which
maps every product state of this system to an entangled bipartite
state if and only if $\min(m,n)\geq 3$ and $(m,n)\neq (3,3)$.
\end{theorem}

There are many literatures and different approaches which are potentially relevant to the topic of entanglers. Besides universal entangler, Zhang \textit{et al.} discussed \textit{perfect} entanglers which are defined as the unitary operations that can generate
maximal entangled states from some initially separable
states~\cite{Zhang04prl,Zhang04pra,Rezakhani04}. Bu\v zek \textit{et al.}~\cite{Buzek00} considered the entangler that entangles a qubit in unknown state with a qubit in a reference, also they proved the nonexistence of universal entangler for qubits, which answers a special case of our problem. Furthermore, nonexistence of universal entangler in $2\otimes n$ bipartite system
can be derived straightforwardly from Parthasarathy's recent work on
the maximal dimension of completely entangled
subspace~\cite{Parthasarathy04}, which is tightly connected to the
unextendible product bases introduced by Bennett \textit{et al.}
~\cite{Bennett99,DiVincenzo03,Bhat04}. 

But the proof of the above
theorem requires some mathematical results from
basic algebraic geometry. We will first review some basic notions in algebraic geometry in the next
section. To make a clear presentation, Theorem~\ref{main} is
divided into two parts, namely Theorems~\ref{thm:degenerate}
and~\ref{thm:general} below, and detailed proofs of them are given.
Finally, a brief conclusion is drawn and some open problems are
proposed. 

\vspace*{3ex}

\paragraph{Preliminaries.}

For the convenience of the reader, we recall some basic definitions in algebraic geometry~\cite{artin91, hartshorne77}.

The \textit{$n\times n$ general linear group} and the \textit{$n\times n$ unitary group} are denoted by $GL(n,\mathbb{C})$ and $U(n)$, respectively.

An \textit{affine $n$-space}, denoted by $A^n$, is the set of all $n$-tuples of complex numbers. An element of $A^n$ is called a point, and if point $P=(a_1,a_2,\cdots,a_n)$ with $a_i\in \mathbb{C}$, then the $a_i$'s are called the coordinates of $P$.

The \textit{polynomial ring in $n$ variables}, denoted by $\mathbb{C}[x_1,x_2,\cdots,x_n]$, is the set of polynomials in $n$ variables with coefficients in a ring.

A subset $Y$ of $A^n$ is an \textit{algebraic set} if it is the common zeros of a finite set of polynomials $f_1,f_2,\cdots,f_r$ with $f_i\in \mathbb{C}[x_1,x_2,\cdots,x_n]$ for $1\leq i\leq r$, which is also denoted by $Z(f_1,f_2,\cdots,f_r)$.

It is not hard to check that the union of a finite number of
algebraic sets is an algebraic set, and the intersection of any
family of algebraic sets is again an algebraic set. Thus by taking
the open subsets to be the complements of algebraic sets, we can define a
topology, called the \textit{Zariski topology} on $A^n$.

A nonempty subset $Y$ of a topological space $X$ is called \textit{irreducible} if it cannot be expressed as the union $Y=Y_1\cup Y_2$ of two proper closed subsets $Y_1$, $Y_2$. The empty set is not considered to be irreducible.

Let $X$, $Y$ be two topological spaces, then we have two useful facts:
\begin{fact}
If $X$ is irreducible and $F: X\rightarrow Y$
be a continuous function, then $F(X)$ with induced topology is also
irreducible.
\end{fact}
\begin{fact}
Let $Y \subset X$ be a subset, if $X$ is irreducible and $Y$ is open, then $Y$ is irreducible; if $Y$ is irreducible and dense in $X$,  then $X$ is irreducible.
\end{fact}

An \textit{affine algebraic variety} is an irreducible closed subset
of some $A^n$, with respect to the induced topology.

We define \textit{projective n-space}, denoted by $\mathbb{P}^n$, to be the set of
equivalence classes of $(n+1)-$tuples $(a_0,\cdots,a_n)$ of complex numbers, not all zero, under the equivalence relation given by
$(a_0,\cdots,a_n)\sim(\lambda a_0,\cdots,\lambda a_n)$ for all
$\lambda \in \mathbb{C}$, $\lambda\neq 0$.

A notion of algebraic variety may also be introduced in projective
spaces, called projective algebraic variety: a subset $Y$ of $\mathbb{P}^n$
is an \textit{algebraic set} if it is the common zeros of a finite
set of homogeneous polynomials $f_1,f_2,\cdots,f_r$ with $f_i\in
\mathbb{C}[x_0,x_1,\cdots,x_n]$ for $1\leq i\leq r$. We call open subsets of irreducible projective varieties as quasi-projective varieties.

We will mainly use the following two varieties:

\begin{example}\label{example:A1}
 $A^1$ is irreducible, because its only proper closed subsets are finite, yet it is itself infinite.
 This fact is somewhat trivial. If $T$ is an algebraic subset of $A^1$, then we should say,
 there is a finite set $F$ of polynomials over $\mathbb{C}$, such that $T=Z(F)$. For all $f\in \mathbb{C}[x]$, if
 $degree(f)$ is 0 or $f$ is a constant function, then $T$
 should be the degenerate subset of $\mathbb{C}$, i.e.,empty set $\emptyset$ or $\mathbb{C}$. Otherwise,
 we can choose a $f \in \mathbb{C}[x]$, such that $f$ is not a constant function, and $degree(f)$ is no less than 1.
 From the fundamental theorem of algebra, we know the number of its roots is at most its degree,
 so its solution set should be a finite set. Thus, a subset of the solution set must be finite too,
 and we proved that $T$ is a finite subset of $\mathbb{C}$.
\end{example}

\begin{example}\label{example:segre}
The \textit{Segre embedding} is defined as the map:
\begin{displaymath}\sigma: \mathbb{P}^{m-1} \times \mathbb{P}^{n-1} \rightarrow
\mathbb{P}^{mn-1}\end{displaymath} taking a pair of points $([x],[y])\in
\mathbb{P}^{m-1}\times \mathbb{P}^{n-1}$ to their product
\begin{eqnarray}
\sigma: ([x_0:x_1:\cdots:x_{m-1}],[y_0:y_1:\cdots:y_{n-1}])\nonumber \\
\longmapsto [x_0y_0:x_0y_1:\cdots:x_{m-1}y_{n-1}]\nonumber
\end{eqnarray}
Here, $\mathbb{P}^{m-1}$ and $\mathbb{P}^{n-1}$ are projective vector spaces over some
arbitrary field,
\begin{displaymath}
[x_0:x_1:\cdots:x_{n-1}]
\end{displaymath}
is the homogeneous coordinates of $x$, and similarly for $y$. The image of the map
is a variety, called \textit{Segre variety}, written as $\Sigma_{m-1,n-1}$.
\end{example}

What concerns us is that Segre variety represents the set of product states~\cite{Brody01, Miyake03, Heydari05}.

If $X$ is a topological space, we define the \textit{dimension of X}, denoted
by $\dim(X)$, to be the supremum of all integers $n$ such that there
exists a chain $Z_0\subset Z_1\subset \cdots\subset Z_n$ of n+1 distinct
irreducible closed subsets of X. The dimension of a quasi-projective variety is then defined according to the Zariski topology.

\vspace*{3ex}

\paragraph{Main results.}
With the notations introduced above, Problem~\ref{or-prob} can be
restated as follows:

\begin{problem}\label{mod-prob} Let Z denote the Segre variety $\Sigma_{m-1,n-1}$. Whether there exists a gate $\Phi\in U(mn)$ s.t.
$\Phi(Z)\cap Z=\emptyset$?\end{problem}

Note that $Z$ is the set of product states, and thus $\Phi(Z)$ is the set of states generated by gate $\Phi$ from product states. So, $\Phi(Z)\cap Z=\emptyset$ means that all states in $\Phi(Z)$ are entangled, and the above problem coincides with Problem~$\ref{or-prob}$.

We claim that the answer to the above question is
affirmative except for some degenerated cases. First, we consider the degenerated cases that $\min(m,n)\leq 2$ or $(m,n)=(3,3)$. The following theorem gives a negative answer for this case:

\begin{theorem}\label{thm:degenerate}
Given a bipartite quantum system $\mathcal{H_m\otimes H_n}$, where $\mathcal{H_m}$ and $\mathcal{H_n}$ are Hilbert spaces with dimension m and n respectively,
if $\min{(m,n)}\leq 2$ or $(m,n)=(3,3)$, then $\forall \Phi\in U(mn)$, we have
\begin{displaymath}
\Phi(Z)\cap Z\neq \emptyset
\end{displaymath}In other words, universal entangler does not exist.
\end{theorem}

To prove the above theorem, we need the following:

\begin{lemma}\label{lemma:proj}[Projective Dimension Theorem]\label{lemma:dim}
Let $Y$, $Z$ be varieties of dimensions r, s in $\mathbb{P}^n$. Then every
irreducible component of $Y\cap Z$ has dimension $\geq r+s-n$.
Furthermore, if $r+s-n\geq 0$, then $Y\cap Z$ is nonempty.
\end{lemma}

For the proof of Lemma~\ref{lemma:proj}, see~\cite{hartshorne77}, chapter
I, theorem 7.2. Then we can easily prove
Theorem~\ref{thm:degenerate}:

\begin{proof}
For any $\Phi\in U(mn)$, $\dim(Z)=m-1+n-1=m+n-2$. Then we have
$\dim(\Phi(Z))+\dim(Z)-\dim(Y)=2\dim(Z)-\dim(Y)=2(m+n-2)-(mn-1)=1-(m-2)(n-2)\geq
0$. So we have $\Phi(Z)\cap Z\neq \emptyset$
according to Lemma~\ref{lemma:proj}.
\end{proof}

We now turn to consider the general case where $\min(m,n)\geq 3$ and $(m,n)\neq (3,3)$.

The following lemmas are needed in the proof of our main theorem.
\begin{lemma}\label{lemma:dense}
U(k) is Zariski dense in GL(k,$\mathbb{C}$)
\end{lemma}
\begin{proof}
For the case $k=1$, the only Zariski closed subsets of
$G=GL(1,\mathbb{C})=\mathbb{C}\backslash \{0\}\cong \{(z_1,z_2)\in \mathbb{C}^2:z_1z_2=1\}$ are
$\emptyset$, $G$ and non-empty finite subsets of G. This follows
immediately from Example~\ref{example:A1}. Since $GL(1,\mathbb{C})$ is an
open subset of affine line $\mathbb{C}$, its closed subsets are intersections
of closed subsets of $\mathbb{C}$ and $\mathbb{C}\backslash \{0\}$. Now we see that $U(1)$ is infinite, thus
$\overline{U(1)}=G$.

In general, 
$H(k)=\{{\rm diag}(z_1,z_2,\cdots,z_k):z_1,z_2,\cdots,z_k\in \mathbb{C}\backslash \{0\}
\}
\cong (\mathbb{C}\backslash \{0\})^k$ is the Fibre product of $k$ copies of $\mathbb{C}\backslash
\{0\}$~\cite{hartshorne77}. Then $H(k)\cap
U(k)=\{{\rm diag}(z_1,z_2,\cdots,z_k):|z_1|=|z_2|=\cdots=|z_k|=1\}\cong
U(1)^k$.

By the results of the case $k=1$, we also have $\overline{H(k)\cap
U(k)}=H(k)\supseteq
A(k)=\{{\rm diag}(z_1,z_2,\cdots,z_k):z_1>0,z_2>0,\cdots,z_k>0\}$.

Now for any $B\in U(k)$, $L_B:X\rightarrow B\cdot X$ and
$R_B:X\rightarrow X\cdot B$ are two isomorphisms of G, and
$L_B(U(k))=U(k)=R_B(U(k))$. Thus
$L_B(\overline{U(k)})=\overline{U(k)}=R_B(\overline{U(k)})$.

Since $A(k)\subseteq \overline{H(k)\cap U(k)}\subseteq
\overline{U(k)}$, we have $U(k)A(k)U(k)\subseteq \overline{U(k)}$.

By singular value decomposition for $GL(k,\mathbb{C})$, we get
$U(k)A(k)U(k)=GL(k,\mathbb{C})\subseteq \overline{U(k)}$. Thus
$\overline{U(k)}=GL(k,\mathbb{C})$.
\end{proof}

The following lemma~\cite{tauvel05} establishes a connection between
the dimensions of domain and codomain of a variety morphism. A
morphism $\Phi:Z_1\rightarrow Z_2$ is called a \textit{dominant}
morphism if $\Phi(Z_1)$ is dense in $Z_2$.

\begin{lemma}\label{lemma:dim}
\noindent
\begin{enumerate}
\item[1]\label{lemma:dim1}
$Z_1$ and $Z_2$ are both irreducible varieties over $\mathbb{C}$, and
$\phi:Z_1\rightarrow Z_2$ is a dominant morphism, then $\dim(Z_2)\leq
\dim(Z_1)$. 
\item[2] \label{lemma:dim2} 
$Z_1$ and $Z_2$ are both
varieties over $\mathbb{C}$, and $\phi:Z_1\rightarrow Z_2$ is a morphism,
$r=\max\limits_{z\in Z_2}{\dim(\phi^{-1}(z))}$, then $\dim(Z_1)\leq
r+\dim(Z_2)$.
\item[3] \label{lemma:cover}
If $V=\cup_{i=1}^s V_i$ is a finite open covering, and $\forall i$, $V_i$ is irreducible, $\forall i,j$, $V_i\cap V_j\neq \emptyset$, then $V$ is irreducible.
\end{enumerate}
\end{lemma}

Let $X=\{\Phi|\Phi\in GL(mn,\mathbb{C}),\Phi(Z)\cap Z\neq \emptyset\}$.
We are able to give an upper bound on the dimension of the closure of $X$ with respect to the Zariski topology.

\begin{lemma}\label{lemma:closure}
$\dim(\overline{X})\leq m^2n^2-(m-2)(n-2)+1$, where $\overline{X}$ is the Zariski
closure of $X$.
\end{lemma}
\begin{proof}
We have a morphism $F:G\times Y\rightarrow Y$ which is just the left
action of $G$ on $Y$, defined by
\begin{displaymath} F(g,[w])=[g\cdot w]
\end{displaymath}
where $G=GL(k,\mathbb{C})$, $Y=\mathbb{P}^{mn-1}$, $k=mn$.

Let $y_0=(1,0,\cdots,0)^T$ be a column vector with $k$ entries. For any
given $y_1$, $y_2\in Y$, we choose $g_1$, $g_2\in
GL(k,\mathbb{C})$, such that $[g_1\cdot y_0]=[y_1]$ and $[g_2\cdot
y_0]=[y_2]$. Then we have
\begin{align*}
 &[g\cdot y_1]=[y_2] \\
\iff & [g g_1\cdot y_0]=[g_2\cdot y_0] \\
\iff & [g_2^{-1}gg_1\cdot y_0]=[y_0]
\end{align*}

From above observations, $F$ has the following property: for any $y_1$, $y_2\in Y$,
$F^{-1}(y_2)\cap \{G\times \{y_1\}\}\cong \{\left(\begin{array}{cc}
z_1& \alpha\\
0 & g'
\end{array}\right):z_1\in \mathbb{C}\backslash \{0\}, g'\in GL(k-1,\mathbb{C}), \alpha\in \mathbb{C}^{k-1} \ is \ a \ row \
vector.\}$. Hence $\dim(F^{-1}(y_2)\cap {G\times
\{y_1\}})=m^2n^2-(mn-1)$.

Let $P_1$, $P_2$ be projections of $G\times Y$ to $G$, $Y$ respectively.
Now we only look at $G\times Z\subseteq G\times Y$, to get
$F:G\times Z\rightarrow Y$. Then we have a characterization of X:
$X=P_1F^{-1}(Z)$. In fact,
\begin{align*}
& g\in X \\
\iff & g(Z)\cap Z\neq \emptyset  \\
\iff &\exists z_1, z_2\in Z, s.t. g(z_1)=z_2 \\
\iff &\exists z_1, z_2\in Z, s.t. (g, z_1)\in F^{-1}(z_2)\\
\iff &\exists z_2\in Z, s.t. g\in P_1F^{-1}(z)\\
\iff & g\in P_1 F^{-1}(Z)
\end{align*}
Let $\overline{X}\subseteq G$ be the Zariski
closure of $X$ in $G$, then $P_1: F^{-1}(Z)\rightarrow \overline{X}$ is a dominant morphism.

Furthermore, consider $\Psi: F^{-1}(Z)\rightarrow Z\times Z$ given by
\begin{displaymath}
\Psi(g,[z])=([z],[g\cdot z])
\end{displaymath}
For all $z_1$, $z_2\in Z$, we have $\Psi^{-1}(z_1,z_2)= (g_2 T
g_1^{-1}, z_1)$, where $T=\{\left(\begin{array}{cc}
z_0& \alpha\\
0 & g'
\end{array}\right):z_0\in \mathbb{C}\backslash \{0\}, g'\in GL(k-1,\mathbb{C}), \alpha\in \mathbb{C}^{k-1} \ is \ a \ row \
vector\}$, and $g_1, g_2\in GL(k,\mathbb{C})$, s.t. $g_1(y_0)=z_1$, $g_2(y_0)=z_2$.
$\Psi$ is a dominant morphism since $G$ acts transitively on $\mathbb{P}^{k-1}$. Then we obtain
\begin{align*}
 \dim(F^{-1}(Z))
\leq & \dim(T)+\dim(Z\times Z)\\
=& m^2n^2-mn+1+2\cdot \dim(Z)
\end{align*}

It is required in Lemma~\ref{lemma:dim1}.1 that varieties $Z_1$ and
$Z_2$ are irreducible, but we haven't proved $F^{-1}(Z)$ is irreducible. Actually, this condition can be weakened:
Lemma~\ref{lemma:dim1}.1 is still true for the more general case that
$Z_1$ and $Z_2$ are closed subsets of irreducible
varieties~\cite{hartshorne77}. Thus, we can fill out the gap and
apply this lemma. Indeed, the
irreducibility of $F^{-1}(Z)$ and $\overline{X}$ really holds, but the proof
is not easy (see~\cite{irreducible} for a brief proof ). From
Lemma~\ref{lemma:dim1}.1, we have
\begin{align*}
\dim(\overline{X})\leq &\dim(F^{-1}(Z))\\
\leq &
(m^2n^2-(mn-1))+2\cdot \dim(Z)\\
=& m^2n^2-(mn-1)+2(m+n-2)\\
=& m^2n^2-(m-2)(n-2)+1
\end{align*}
\end{proof}

With the above lemmas, we can now derive the main result:
\begin{theorem}\label{thm:general}
Given a bipartite quantum system $\mathcal{H_m\otimes H_n}$, where $\mathcal{H_m}$ and $\mathcal{H_n}$ are Hilbert spaces with dimensions $m$ and $n$ respectively, if  $\min{(m,n)}> 2$ and $(m,n)\neq(3,3)$, then there exists an
unitary operator $\Phi\in U(mn)$, s.t.
\begin{displaymath}
\Phi(Z)\cap Z=\emptyset
\end{displaymath}That is, $\Phi$ is a universal entangler.
\end{theorem}
\begin{proof}
If $U(mn)\subset X$, then it follows that
$GL(mn,\mathbb{C})=\overline{U(mn)}\subset{\overline{X}}$ from
Lemma~\ref{lemma:dense}. And in this assumption, we also have
$\dim(\overline{X})\leq m^2n^2-(m-2)(n-2)+1 < m^2n^2 = \dim(GL(mn,\mathbb{C}))$
 from Lemma~\ref{lemma:closure}~\cite{further}. It's a contradiction.
So $U(mn)\not\subset X$, i.e. a unitary operator $\Phi\in U(mn)$
with $\Phi(Z)\cap Z= \emptyset$ exists.
\end{proof}

\vspace*{3ex}

\paragraph{Conclusion.}
In summary, it is shown that a universal entangler of bipartite
system exists except for the cases of $1\otimes n$, $2\otimes n$,
$m\otimes 1$, $m\otimes 2$ and $3\otimes 3$. So we have completely
determined when a universal entangler exists. It seems that the
method employed in this Letter can be extended to the multipartite
case. This extension will lead us to
explore the geometric structure of entangled states and other
related objects. Furthermore, we can consider some more practical
questions: (1) How to construct such a universal entangler for
$3\otimes 4$ bipartite system explicitly? (2) Given a universal entangler, what is the minimum entanglement it guarantees to output with respect to the definition of entanglement measure for pure states~\cite{Bennett96,Bennett97,Popescu97}? (3) Furthermore, what is the optimal universal entangler which maximized the minimally possible entanglement the entangler outputs. Intuitively, for a bipartite system of sufficiently large dimensions,a randomly chosen unitary operator seems to be a universal entangler with high probability. Nevertheless,
we failed to give a proof of this conjecture.

\vspace*{3ex}

\paragraph{Acknowledgements.}
We are thankful to the colleagues in the Quantum Computation and
Information Research Group of Tsinghua University for helpful discussions. And one of us (J.-X. Chen) enjoyed
delightful discussions with W. Huang, L. Xiao, Z. Q. Zhang and Q.
Lin. Jun Yu thank Professor X.-J. Tan
and Q.-C. Tian who ask him to study Hartshorne 's book. This work was partly supported by the Natural Science
Foundation of China(Grant Nos.60621062 and 60503001) and the Hi-Tech
Research and Development Program of China(863 project)(Grant
No.2006AA01Z102).

\bibliography{unitary}

\begin{thebibliography}{10}
\bibitem{irreducible}
Here, we will present a short proof of the irreduciblity of $F^{-1}(Z)$.
  $\forall i,j$: $0\leq i\leq {m-1}$, $0\leq j\leq {n-1}$, let $P_{m,i}=\{[x_0,
  x_1, \cdots, x_i,\cdots,x_{m-1}]|x_i\neq 0\}$, and $P_{n,j}=\{[y_0, y_1,
  \cdots, y_j,\cdots,y_{n-1}]|y_i\neq 0\}$. Let $Z_{i,j}=\sigma(P_{m,i}\times
  P_{n,j})$, we will have $Z=\cup_{i,j} Z_{i,j}$, and for $\forall i,j$,
  $Z_{i,j}\cong P_{m,i}\times P_{n,j}\cong A^{m-1}\times A^{n-1} \cong
  A^{m+n-2}$, and each $Z_{i,j}$ is open in Z. $\forall i,j$, we can construct
  a morphism $T_{i,j}: Z_{i,j}\rightarrow GL(mn,C)$ s.t. $\forall y\in
  Z_{i,j}$, $(T_{i,j}(y))(y_0)=y$. Then $\forall i,j,\hat{i},\hat{j}$,
  $(g,y)\in \Psi^{-1}(Z_{\hat{i},\hat{j}}\times Z_{i,j}) \iff g\cdot y\in
  Z_{i,j}\ and\ y\in Z_{\hat{i},\hat{j}}\iff (g\cdot
  T_{\hat{i},\hat{j}}(y_1))y_0=T_{i,j}(y_2)(y_0)\ for\ some\ y_2\in Z_{i,j},\
  y_1\in Z_{\hat{i},\hat{j}}\iff g\in T_{i,j}(Z_{i,j})\cdot H\cdot
  (T_{\hat{i},\hat{j}}(Z_{\hat{i},\hat{j}}))^{-1}$. Thus we have
  $F^{-1}(Z)_{i,j,\hat{i},\hat{j}}=\{(g,y)\in
  \Psi^{-1}(Z_{\hat{i},\hat{j}}\times Z_{i,j})\}$ is irreducible since it's the
  image of a morphism $Z_{i,j}\times H\times Z_{\hat{i},\hat{j}}\rightarrow
  F^{-1}(Z)$. Each $F^{-1}(Z)_{i,j,\hat{i},\hat{j}}$ is also open in
  $F^{-1}(Z)$, then the irreduciblity follows from lemma 3.3.
\bibitem{further}
In fact, we proved $\Psi:F^{-1}(Z)\rightarrow Z\times Z$ is an algebraic fibre
  bundle thus $dim (F^{-1}(Z))=m^2n^2-mn+1+2\cdot dim(Z)$. We also conjecture
  that $dim(\overline{X})=dim(F^{-1}(Z))$ and $X=\overline{X}$.
\end{thebibliography}
\bibliographystyle{plain}
\end{document}